\documentclass[12pt,a4paper, twoside,reqno]{amsart}
\usepackage{amscd}
\usepackage{amsmath,amssymb,amsthm,mathrsfs}
\usepackage[margin=3cm]{geometry}
\usepackage{tikz}
\usepackage{subfig}
\usepackage{accents}
\usetikzlibrary{calc}
\usetikzlibrary{matrix,arrows,decorations.pathmorphing,decorations.markings,calc}

\usepackage{parskip}
\usepackage{enumerate}
\newtheorem{theorem}{Theorem}[section]

\newtheorem{lem}[theorem]{Lemma}

\theoremstyle{definition}
\newtheorem{definition}[theorem]{Definition}

\theoremstyle{remark}

\numberwithin{equation}{section}
\title{Commutativity in Lagrangian and Hamiltonian Mechanics}

\author{Ananth Sridhar and Yuri B. Suris}
\email{asridhar@berkeley.edu, suris@math.tu-berlin.de}
\begin{document}

\maketitle

\begin{abstract}
The main result of this note is a characterization of the Poisson commutativity of Hamilton functions in terms of their principal action functions. 
\end{abstract}
\noindent
\section{Introduction}
Let $ M $ be a finite dimensional manifold, and $ L_1,  L_2 \in C^\infty(TM) $ two non-degenerate Lagrange functions (with locally invertible Legendre transformations). Let $ S_1, S_2$ be their principal action functions (see section \ref{sect cont results} for definitions and details), and  let $ H_1, H_2 \in C^\infty(T^*M) $ be the corresponding Hamilton functions on the phase space $ T^*M $ with its canonical Poisson bracket.
\begin{theorem}\label{thm:thm1}
If the principal action functions satisfy
\begin{align} \label{eq:coms}
\min_{q_1\in M} \Big(S_1( q_0,q_1, t_1) + S_2(q_1, q_{12}, t_2)\Big) = \min_{q_2\in M} \Big(S_2( q_0, q_2, t_2) + S_1(q_2,q_{12}, t_1)\Big)
\end{align}
for all $ (q_0, q_{12})$ from some neighborhood of the diagonal in $M\times M$ and for sufficiently small $ t_1, t_2 > 0 $, then the Hamilton functions Poisson commute, $ \{ H_1, H_2 \} = 0 $.
 \end{theorem}

We find a discrete time counterpart of Theorem \ref{thm:thm1} very instructive and enlightening, concerning both the statement and the proof. Let $\Lambda_1,\Lambda_2\in C^\infty(M\times M)$ be two discrete time Lagrange functions, generating two symplectomorphisms $F_1$, $F_2$ of $T^*M$ (see section \ref{sect: discr results} for definitions and details).
\begin{theorem}\label{thm:thm2}
If the discrete time Lagrange functions satisfy
\begin{align} \label{eq:coms discr}
\min_{q_1\in M} \Big(\Lambda_1( q_0,q_1) + \Lambda_2(q_1, q_{12})\Big) = \min_{q_2\in M} \Big(\Lambda_2( q_0, q_2) + \Lambda_1(q_2,q_{12})\Big)
\end{align}
for all $ (q_0, q_{12})$ from some neighborhood of the diagonal in $M\times M$, then the symplectic maps $F_1$, $F_2$ commute.
 \end{theorem}
 
 Our interest in the Lagrangian characterization of commutativity arose from studying two seemingly unrelated areas:  integrability of the semiclassical limits of quantum integrable systems and of solvable lattice models, on the one hand, and pluri-Lagrangian calculus, on the other hand. Let us briefly outline this motivation.

\paragraph{\em Quantum Integrable Systems and Solvable Lattice Models}
Classical integrable systems are most naturally described within the Hamiltonian framework. Here the phase space is a $ 2n $-dimensional symplectic manifold $ (N, \omega) $. The dynamics of the system is determined by a Hamilton function $ H \in  C^\infty(N) $ and Hamiltonian equations for the time evolution of phase space functions,
\begin{align*}
\frac{df}{dt} = \{H, f\},
\end{align*}
where the Poisson bracket is given by $ \{ f, g \} = \omega^{-1}(df, dg)  $. A system is \emph{Liouville integrable} if it admits $ n $ functionally independent functions $  H_1, H_2, \ldots, H_n $ all pairwise Poisson commuting, $ \{ H_i, H_j \} = 0 $.

In a quantum mechanical system, the phase space is replaced by a Hilbert space of states $ \mathfrak{H} $, and the Hamiltonian by a Hermitian operator $ \hat{H}: \mathfrak{H} \rightarrow \mathfrak{H} $. The time evolution of an observable $ \hat{O} $ in the Heisenberg picture is given by
\begin{align*}
\frac{d \hat{O}}{dt} = \frac{i}{\hbar} [\hat{H}, \hat{O}],
\end{align*} 
where $ [\cdot,\cdot] $ is the commutator in the algebra of observables. In analogy to classical integrable systems, a \emph{quantum integrable system} can be characterized by the existence of ``many'' commuting Hamiltonians.

The semiclassical limit of a quantum mechanical system is found by taking $ \hbar \rightarrow 0 $. In this limit, the commutativity of quantum Hamiltonians and the correspondence principle, 
\begin{align*}
[ \hat{H}_1, \hat{H}_2 ] = i \hbar \{H_1,H_2 \} + O(\hbar^2),
\end{align*}
imply the Poisson commutativity of classical Hamiltonians.

Theorem \ref{thm:thm1} gives an alternative proof of commutativity of classical Hamiltonians for semiclassical quantizations. For example, in one spatial dimension, the kernel $ U(q_1,q_2,t) $ of the propagator $ e^{i \hat{H} t / \hbar} $ has the WKB asymptotic 
\begin{align}\label{eq:wkb}
U(q_1,q_2,t) = \left( \frac{1}{2 \pi i \hbar} \, \frac{\partial^2 S }{ \partial q_1  \partial q_2 } (q_1, q_2, t) \right)^{-\frac{1}{2} } e^{\frac{i}{\hbar} S(q_1,q_2,t) } \left( 1 + O(\hbar) \right),
\end{align}
as $ \hbar \rightarrow 0 $,  where $ S $ is the principal action function of the classical system.

The commutativity of Hamiltonians $ \hat{H}_1 $ and $ \hat{H}_2 $ implies the commutativity of propagators, which can in turn be expressed in terms of their integral kernels as
\begin{align*}
\int_\mathbb{R}  U_1(q_0,q,t_1) \, U_2(q,q_{12},t_2) \, dq = \int_\mathbb{R} \, U_2(q_0,q,t_2) \,  U_1(q,q_{12},t_1) \, dq,
\end{align*}
for any $ t_1, t_2 $. In the semiclassical limit $ \hbar \rightarrow 0 $, by substituting the asymptotic (\ref{eq:wkb}) and applying the stationary phase approximation, it can be seen that the commutativity of the quantum Hamiltonians leads to relation (\ref{eq:coms}) for their actions. Theorem \ref{thm:thm1} then implies that the classical Hamiltonians also commute.

A different application of Theorem \ref{thm:thm1} is found in the thermodynamic limit of certain lattice models; we refer the reader to \cite{Res10,RS} for details. We consider as an example vertex models solvable by the transfer matrix method. Yang-Baxter equations for parameters of these models ensure that the transfer matrices form commutative families of operators. This is an analog to the classical Liouville integrability.


Many lattice models exhibit the \emph{limit shape phenomenon}, see \cite{CKP,Ok,KOS,ZJ}: as the number of sites $\mathcal N$ grows, the macrostate of the system becomes deterministic with the thermodynamic fluctuations becoming exponentially small  in $\mathcal  N $. 

On a cylinder $\mathbb Z_N\times [1,M]$, states are assigned to $\mathbb Z_N$ slices. The thermodynamic limit is found by taking $ N,M \rightarrow\infty $, with the aspect ratio $ t = M/N $ fixed.  A semiclassical state $ \nu: S^1 \rightarrow \mathbb{R} $ is a limit of a convergent (in a certain sense) sequence of states,  $\nu=\lim_{N\to \infty}v^{(N)} $. The limit shape of the system can be determined by a variational principle as follows. 
Fixing two semiclassical states $ \nu_1 =\lim_{N\to\infty}v_1^{(N)}$ and $ \nu_2=\lim_{N\to\infty}v_2^{(N)} $, on the top and on the bottom of the cylinder, the matrix elements of the transfer matrix (i.e., the partition functions) have the asymptotic
\begin{align} \label{eq:transfasym}
\lim_{\substack{N,M \rightarrow \infty \\  M/N = t }} (T^N)_{v_1^{(N)} v_2^{(N)} } = e^{ -N^2 S(\nu_1, \nu_2, t) }\left( 1 + o(1) \right),
\end{align}
where $ S $ is the principal action functional of the field theory
\begin{align*}
 S(\nu_1, \nu_2, t) = \min_{\substack{\varphi:\ S^1 \times [0,t]\to\mathbb R \ \\ \varphi(x,0) = \nu_1(x) \\ \varphi(x,t) =\nu_2(x) }}\int_0^t \int_{S_1} \sigma( \partial_x \varphi, \partial_t \varphi ) \, dx \, dt.
\end{align*}
The minimization is over functions satisfying certain conditions on the gradient we do not specify here. The function $ \sigma $ is the surface tension \cite{CKP, PR}. 
The minimizer $ \varphi^* $ of this variational problem is called the limit shape of the system. It is uniquely defined if  $ \sigma $ is strictly convex.

In the limit $ N \rightarrow \infty $, using the asymptotic (\ref{eq:transfasym}) along with the stationary phase approximation, it is seen that the commutativity of transfer matrices leads to the relation (\ref{eq:coms}) for their actions. Theorem \ref{thm:thm1} suggests that after the Lagrangian field theory is reformulated as a Hamiltonian field theory, the corresponding Hamiltonians $ H_1 $ and $ H_2 $ Poisson commute. For the six vertex model, the Poisson commutativity was proven directly in \cite{RS}.

\paragraph{\em Pluri-Lagrangian systems.} The converse statements to Theorems \ref{thm:thm1} and \ref{thm:thm2} are also easily shown. This relates those theorems to the pluri-Lagrangian theory of commuting Hamiltonian flows and commuting symplectic maps \cite{Sur1}. If the flows $F_1^{t_1}$ and $F_2^{t_2}$ of the Hamilton functions $H_1$ and $H_2$ commute, then, for any $(q_0,p_0)\in T^*M$, one can define a function $(q,p):\mathbb R^2\to T^*M$ by setting
$$
\big(q(t_1,t_2),p(t_1,t_2)\big)=F_1^{t_1}\circ F_2^{t_2}(q_0,p_0).
$$
The functions $q:\mathbb R^2\to M$ obtained by the canonical projection $T^*M\to M$ have the following remarkable property: for any curve $\gamma$ in $\mathbb R^2$ connecting $(0,0)$ with $(t_1,t_2)$, the restriction $q\circ \gamma$ minimizes the action functional $S_\gamma=\int_\gamma (L_1dt_1+L_2dt_2)$ under the boundary conditions $q(0,0)=q_0$ and $q(t_1,t_2)=q_{12}$; moreover, if $\{H_1,H_2\}=0$ then the critical value of $S_\gamma$ does not depend on $\gamma$. Theorem \ref{thm:thm1} effectively deals with $S_{\gamma_1}$ and $S_{\gamma_2}$ for two stepped curves
$$
\gamma_1(t)=\left\{\begin{array}{ll} (t,0) &{\rm for}\;\; t\in[0,t_1],\\
                                                        (t_1,t-t_1) &{\rm for}\;\; t\in[t_1,t_1+t_2],\end{array}\right. 
$$                                                        
resp.
$$
\gamma_2(t)=\left\{\begin{array}{ll} (0,t) & {\rm for}\;\; t\in[0,t_2],\\
                                                        (t-t_2,t_2) & {\rm for}\;\; t\in[t_2,t_1+t_2],\end{array}\right.
$$
and claims that equality of the critical values of $S_{\gamma_1}$ and $S_{\gamma_2}$ is sufficient for the validity of the pluri-Lagrangian picture. For a general discussion of the role of stepped curves (resp. surfaces) in the pluri-Lagrangian calculus see \cite{SV}.

\section{Commuting Actions and Hamiltonians}
\label{sect cont results}

\subsection{Lagrangian and Hamiltonian Mechanics} 
\label{subsect Lagr Ham}

We first review some standard concepts to fix notation and definitions. 
Let $M$ be  a smooth, compact, $n$-dimensional manifold. It will play the role of the configuration space. Let $ L: TM \rightarrow \mathbb{R} $, $ (q,\dot q) \mapsto L(q,\dot q) $ for $ q\in M $ and $ \dot q \in T_q M $  be a smooth Lagrange function. The action functional of a continuous path $q : [t_1,t_2] \rightarrow M $ is
\begin{align}\label{eq:action}
\mathcal{S}[q] = \int_{t_1}^{t_2} L(q(t), \dot q (t)) dt,
\end{align}
where $ \dot q (t) = \frac{d}{dt}q(t) \in T_{q(t)}M $. The classical trajectories of the system are determined by the principle of least action. The trajectory starting at $ q_1 \in M $ at time $ t_1 $  and  arriving at $ q_2 \in M$ at $ t_2 $ minimizes the action \eqref{eq:action} over all continuous paths $q: [t_1, t_2] \rightarrow M$ satisfying $q(t_1) = q_1$ and 
$q(t_2) = q_2$ (we assume that the minimizer exists and is unique). The critical value of action is called {\em principal action function}:
\begin{align} \label{eq:acS}
 S(q_1, q_2, t_2-t_1)  = \text{min} \; \big\{ \mathcal{S}[q] \; \big |  \; q: [t_1, t_2] \rightarrow M, \; q(t_1) = q_1, \; q(t_2) = q_2 \big\}.
\end{align}
This defines a smooth function at least as long as $(q_1,q_2)$ lies in some neighborhood of the diagonal in $ M \times M$ and $t_2-t_1>0$ is sufficiently small.

Classical trajectories satisfy the Euler-Lagrange equations. In local coordinates,
\begin{align}\label{eq:eleq}
\frac{d}{dt} \frac{\partial L }{\partial \dot q^j }( q(t) , \dot q(t) ) - \frac{\partial L}{\partial q^j}( q(t), \dot q(t)  )  = 0, \quad j = 1, \ldots, n.
\end{align}

The formula
\begin{equation}\label{eq:p}
p_j=\frac{\partial L(q,\dot q)}{\partial \dot q^j}, \quad j=1,\ldots,n,
\end{equation}
defines the {\em Legendre transformation} which is a vector bundle map $TM\to T^*M$, $(q,\dot q)\mapsto (q,p)$. Under the standard non-degeneracy conditions,
$$
\det\left(\frac{\partial^2 L}{\partial \dot q^j\partial \dot q^k}\right)_{j,k=1}^n\neq 0,
$$
map \eqref{eq:p} is a local diffeomorphism and possesses (locally) the inverse function $\dot q=\dot q(q,p)$.
The Hamilton function $ H: T^*M \rightarrow \mathbb{R} $ is by definition the Legendre transform of $ L $:
\begin{equation} \label{eq:leg}
H(q, p) = \big( \langle p, \dot q \rangle  - L(q,\dot q)  \big)\big |_{\dot q=\dot q(q,p)},
\end{equation}
here $ \langle p, \dot q \rangle $ is the pairing of $ p \in T^*_q M $ and $ \dot q \in T_q M $. 
%
%

Solutions to the Euler-Lagrange equations are mapped by the Legendre transformation \eqref{eq:p} to solutions of the Hamiltonian flow with Hamilton function $ H $, which is given 
in local coordinates by
\begin{align}\label{eq:hamflow}
\frac{d q^j}{dt} =  \{ H , q^j \}=\frac{\partial H}{\partial p_j},  \quad \frac{d p_j}{dt} = \{ H, p_j \}=-\frac{\partial H}{\partial q^j}, \quad  j = 1, \ldots, n,
\end{align}
where the canonical Poisson bracket of $ f, g \in C^\infty(T^*M ) $ is
\begin{align}
\{f, g \} = \sum_{j = 1}^n \left(\frac{\partial f}{\partial p_j} \frac{\partial g}{\partial q^j} - \frac{\partial f}{\partial q^j} \frac{\partial g}{\partial p_j} \right).
\end{align}
Conversely, solutions of the flow \eqref{eq:hamflow} project to solutions of (\ref{eq:eleq}) by the natural bundle projection $ \pi: T^*M \rightarrow M $. Thus, the Hamiltonian and Lagrangian frameworks are equivalent.

Assuming as before that the principal action function  (\ref{eq:acS}) is a smooth function,,
the {\em Hamilton-Jacobi equations} give the differential of $ S $. In coordinates,
\begin{equation} \label{eq:hje}
\frac{\partial S(q_1,q_2,t)}{\partial q^j_1} =  -p_j(t_1)=-\frac{\partial L}{\partial \dot q^j}(q_1,\dot q_1),\quad
\frac{\partial S(q_1,q_2,t)}{\partial q^j_2} =  p_j(t_2)=\frac{\partial L}{\partial \dot q^j}(q_2,\dot q_2), 
\end{equation}
\begin{equation} \label{eq:hje 2}
\frac{\partial S(q_1,q_2,t)}{\partial t} =   -  H,
\end{equation}
where  $\dot q_1=\dot q(t_1)$ and $\dot q_2=\dot q(t_2)$ are velocities at the endpoints of the minimizing path, and $ H = H(q(t),p(t)) $ is constant along the trajectory.

\subsection{Composing Actions}
Let $ L_1,  L_2 : TM \rightarrow \mathbb{R} $ be two smooth non-degenerate Lagrange functions, and $ S_1, S_2 $ and $ H_1, H_2 $ their corresponding principal action functions and Hamilton functions.

For $ q_0, q_{12} \in M $ and $ t_1, t_2 > 0 $, let
\begin{equation} \label{eq:gS}
S_{12}(q_0,q_{12}, t_1, t_2) = \min_{q_1 \in M} \Big( S_1(q_0, q_1, t_1) +   S_2(q_1,q_{12}, t_2) \Big).
\end{equation}
This function can be understood as follows. Consider the action functional 
\begin{equation} \label{eq:gluedfunctional}
\mathcal{S}_{12}[q] = \int_{0}^{t_1} L_1(q (t), \dot q(t)) dt+\int_{t_1}^{t_1+t_2} L_2(q (t), \dot q(t)) dt,
\end{equation}
defined on continuous paths $q: [0,t_1+t_2] \rightarrow M $ with $ q(0) = q_0 $ and $ q(t_1+t_2) = q_{12} $.
It is clear that minimizing trajectory of this action functional satisfies the Euler-Lagrange equations of $ L_1 $ for $ t \in( 0,t_1) $, and Euler-Lagrange equations of $ L_2 $ for $ t 
\in(t_1,t_1+t_2)$. Minimizing $ \mathcal{S}_{12}[q] $  reduces to minimizing with respect to the intermediate point $ q(t_1)=q_1 $. Thus, $ S_{12}(q_0,q_{12},t_1,t_2) $ is the principal action function of the functional (\ref{eq:gluedfunctional}).

Similarly, we set
\begin{equation} \label{eq:gS 21}
S_{21}(q_0,q_{12}, t_2, t_1) = \min_{q_2 \in M} \Big( S_2(q_0, q_2, t_2) +   S_1(q_2,q_{12}, t_1) \Big),
\end{equation}
which is the principal action function for the functional
\begin{equation} \label{eq:gluedfunctional 21}
\mathcal{S}_{21}[q] = \int_{0}^{t_2} L_2(q (t), \dot q(t)) dt+\int_{t_2}^{t_1+t_2} L_1(q (t), \dot q(t)) dt,
\end{equation}
defined on continuous paths $q: [0,t_1+t_2] \rightarrow M $ such that $ q(0) = q_0 $ and $ q(t_1+t_2) = q_{12} $.

\begin{definition}
We say that the principal actions of the Lagrange functions $L_1$, $L_2$ commute, if the two functions \eqref{eq:gS}, \eqref{eq:gS 21} coincide identically:
\begin{equation}\label{eq:S12=S21}
S_{12}(q_0,q_{12},t_1,t_2)=S_{21}(q_0,q_{12},t_2,t_1)
\end{equation} 
for all $(q_0,q_{12})$ from some neighborhood of the diagonal in $M\times M$ and for all sufficiently small $t_1,t_2>0$. 
\end{definition}

\begin{theorem}\label{thm:main cont}
If the principal actions of the Lagrange functions $L_1$, $L_2$ commute then the corresponding Hamilton functions Poisson commute, $\{H_1,H_2\}=0$.
\end{theorem}

\noindent We first establish a few lemmas used in the proof of the theorem.

\begin{lem} \label{lem:lem1}
The conjugate momentum $p^*$ along the critical curve $q^*$ of the functional \eqref{eq:gluedfunctional} is continuous at $t=t_1$.
\end{lem}
\begin{proof}
Let $ q_1^* $ be the critical value of $ q_1 $ minimizing (\ref{eq:gS}) so that
\begin{align} \label{eq: lem1 aux0}
S_{12}(q_0,q_{12},t_1, t_2) =  S_1(q_0, q_1^*, t_1) +  S_2(q_1^*, q_{12}, t_2).
\end{align}
Due to criticality, we have:
\begin{equation} \label{eq: lem1 aux}
\frac{\partial S_1}{\partial q_1^j}(q_0,q_1^*,t_1) +  \frac{\partial S_2}{\partial q_1^j}(q_1^*, q_{12}, t_2) = 0, \quad j=1,\ldots,n. 
\end{equation}
Using Hamilton-Jacobi equation (\ref{eq:hje}) to differentiate $S_1$, $S_2$ with respect to $ q_1 $, we find:
\begin{equation}
\lim_{t\to t_1-0}\frac{\partial L_1}{\partial \dot q^j}(q^*(t),\dot q^*(t)) -  \lim_{t\to t_1+0}\frac{\partial L_2}{\partial \dot q^j}(q^*(t),\dot q^*(t))  = 0, \quad j=1,\ldots,n,
\end{equation}
or
\begin{equation}
\lim_{t\to t_1-0}p_j^*(t)-  \lim_{t\to t_1+0}p_j^*(t) = 0, \quad j=1,\ldots,n.
\end{equation}
This proves the lemma.
\end{proof}

\begin{lem}
Partial derivatives of the principal action function $ S_{12} : M \times M \times \mathbb{R}_{> 0} \times \mathbb{R}_{>0} \rightarrow\mathbb{R} $ are given by:
\begin{align} 
& \frac{\partial S_{12}(q_0,q_{12},t_1,t_2)}{\partial q_0^j} =-p_j^*(0)=-\frac{\partial L_1}{\partial \dot q^j}(q^*(0),\dot q^*(0)),  \label{eq:genHJE 1} \\
& \frac{\partial S_{12}(q_0,q_{12},t_1,t_2)}{\partial q_{12}^j} =p_j^*(t_1+t_2)=\frac{\partial L_2}{\partial \dot q^j}(q^*(t_1+t_2),\dot q^*(t_1+t_2)), \label{eq:genHJE 2}\\
& \frac{\partial S_{12}(q_0,q_{12},t_1,t_2)}{\partial t_1} =-H_1(q^*(0),p^*(0)), \label{eq:genHJE 3}\\
& \frac{\partial S_{12}(q_0,q_{12},t_1,t_2)}{\partial t_2} =-H_2(q^*(t_1+t_2),p^*(t_1+t_2)). \label{eq:genHJE 4}
\end{align}
\end{lem}
\begin{proof}
We differentiate $S_{12}$ using expression \eqref{eq: lem1 aux0}. Since $ q_1^* $ depends on  $ q$, $q_{12}$, $t_1 $ and $ t_2 $, differentiating with respect to any of these involves differentiation with respect to $q_1^*$ and an application of the chain rule. However, by \eqref{eq: lem1 aux}, all such terms vanish. The remaining terms follow from Hamilton-Jacobi equations \eqref{eq:hje}, \eqref{eq:hje 2}.
\end{proof}

\begin{lem} \label{lem: p p12}
Let $ q^*: [0, t_1+t_2] \rightarrow M $ be the minimizing path corresponding to $S_{12}(q_0,q_{12},t_1,t_2)$, and $ q^{**}:[0, t_1+t_2] \rightarrow M $ the minimizing path corresponding to $S_{21}(q_0,q_{12},t_2,t_1)$. Let $ p^* $ and $ p^{**}$ be the corresponding conjugate momenta. If the principal action functions of the Lagrange functions $L_1$, $L_2$ commute then the phase space trajectories $ p^* $ and $ p^{**} $ have the same endpoints:
\begin{equation}\label{eq:comp}
p^*(0)= p^{**}(0), \quad
p^*(t_1+t_2) = p^{**}(t_1+t_2).
\end{equation}
\end{lem}
\begin{proof}
Use equations (\ref{eq:genHJE 1}), (\ref{eq:genHJE 2}) to compute partial derivatives with respect to $ q_0 $ and $ q_{12} $ of both sides of equation $S_{12}=S_{21}$. Equating the components gives (\ref{eq:comp}).
\end{proof}

\begin{lem} If the principal action functions of the Lagrange functions $L_1$, $L_2$ commute then the flows $F_1^{t_1}$ and $F_2^{t_2}$ of the Hamilton functions $H_1$ and $H_2$ commute.
\end{lem}
\begin{proof}
Take an arbitrary $ (q_0,p_0) \in T^*M $.  Let $ (q_1, p_1)=F_1^{t_1}(q_0,p_0)$ and  $ (q_{12}, p_{12})=F_2^{t_2}(q_1,p_1)$. Due to Lemma \ref{lem:lem1}, the minimizing path $q^*(t)$ of $S_{12}(q_0,q_{12},t_1,t_2)$ is the canonical projection to $M$ of the curve $(q^*,p^*):[0,t_1+t_2]\to T^*M$ defined by
$$
(q^*(t),p^*(t))=\left\{\begin{array}{ll}
F_1^t(q_0,p_0), & 0\le t\le t_1, \\ F_2^{t-t_1}(q_1,p_1), & t_1\le t\le t_1+t_2.
\end{array}\right.
$$ 

Now, with $q_{12}$ as above, consider the minimizing path $q^{**}:[0,t_1+t_2]\to M$ of $S_{21}(q_0,q_{12},t_2,t_1)$, along with its lift $(q^{**},p^{**}):[0,t_1+t_2]\to T^*M$. 
By Lemma \ref{lem: p p12}, we find: $p^{**}(0)=p^*(0)=p_0$ and $p^{**}(t_1+t_2)=p^*(t_1+t_2)=p_{12}$. Setting $ (q_2, p_2)=F_2^{t_2}(q_0,p_0)$, we see that 
$$
(q^*(t),p^*(t))=\left\{\begin{array}{ll}
F_2^t(q_0,p_0), & 0\le t\le t_2, \\ F_1^{t-t_2}(q_2,p_2), & t_2\le t\le t_1+t_2.
\end{array}\right.
$$ 
In particular, $ (q_{12}, p_{12})=F_1^{t_1}(q_2,p_2)$. This proves the lemma. 
\end{proof}

The situation is summarized in Figure \ref{fig:paths}. The segments  I  and  IV  are integral curves of the flow $F_1^t$, while  II  and  III are integral curves of the flow $F_2^t$. 

\begin{figure}[h]
\begin{tikzpicture}[scale=1]

\draw (-3,-.2) -- (-3,3);
\node at (-3,3.15){ \tiny{p}};
\draw (-3.2,0)--(3,0);
\node at (3.15,0){ \tiny{q}};

\draw plot [smooth] coordinates {(-2,1) (-1,.8) (-.2,.5)};
\draw plot [smooth] coordinates {(-.2,.5)(1.3,1)(2,2)};
\draw plot [smooth] coordinates {(-2,1)(-1.6,1.5)  (.3,2.6)};
\draw plot [smooth] coordinates {(.3,2.6) (1,2.3) (2,2)};

 \node at (-.7,1.8){\tiny { III }};
 \node at (.9,2){\tiny { IV  }};

 \node at (-1.2,1){\tiny {  I }};
 \node at (.8,1){\tiny {  II }};

\fill (-.2,.5)  circle (2pt); \node at (.2,.2){\tiny{ $(q_1,p_1)$ } };
\fill (.3,2.6) circle (2pt); \node at (.7,2.9){\tiny{ $(q_2,p_2)$ } };

\fill (-2,1) circle (2pt); \node at (-2.4,1.3){\tiny{ $ (q_0, p_0) $ } };
\fill (2,2) circle (2pt);\node at (2.4,2.3) {\tiny{ $ (q_{12}, p_{12}) $ } };

\end{tikzpicture}
\caption{The picture in phase space. Curves I and IV are integral curves of the flow $F_1^t$, curves II and III are integral curves of the flow $F_2^t$. Curves I and II build the path $(q^ *,p^*)$, curves III and IV build the path $(q^{**},p^{**})$.} \label{fig:paths}
\end{figure}
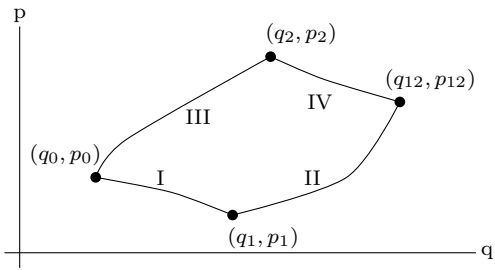 

\begin{lem} If the principal action functions of the Lagrange functions $L_1$, $L_2$ commute then
\begin{align}
H_1(q_0,p_0) = H_1(q_{12},p_{12}), \label{eq:comh 1}\\
H_2(q_0,p_0) = H_2(q_{12},p_{12}). \label{eq:comh 2}
\end{align}
\end{lem}
\begin{proof}
Indeed, both parts of equation \eqref{eq:comh 1} represent, according to (\ref{eq:genHJE 3}), (\ref{eq:genHJE 4}), partial derivatives
$\partial S_{12}/\partial t_1$ resp. $\partial S_{21}/\partial t_1$. Similarly, both parts of equation \eqref{eq:comh 2} represent partial derivatives
$\partial S_{21}/\partial t_2$ resp. $\partial S_{12}/\partial t_2$. 
\end{proof}

\noindent We now return to Theorem \ref{thm:main cont}.

{\em Proof of Theorem \ref{thm:main cont}.}
 From invariance of the Hamilton function along integral curves of a Hamiltonian flow we get:
\begin{align} \label{eq:midpoints}
H_1(q_1,p_1) = H_1(q_0,p_0), \quad H_2(q_{12},p_{12})=H_2(q_1,p_1),\\
H_2( q_2,p_2) = H_2(q_0,p_0), \quad H_1(q_{12},p_{12})=H_1(q_2,p_2).
\end{align}
Comparing this with \eqref{eq:comh 1}, \eqref{eq:comh 2}, we conclude that
\begin{align} \label{eq:midpoints}
H_1(q_1,p_1) = H_1(q_{12},p_{12}), \quad H_2(q_0,p_0)=H_2(q_1,p_1),\\
H_2(q_2,p_2) = H_2(q_{12},p_{12}), \quad H_1(q_0,p_0)=H_1(q_2,p_2).
\end{align}
Thus, $H_1$ is invariant along the integral curves of the flow of $H_2$ (curves II and III in Figure \ref{fig:paths}), and vice versa, so that
$\{H_1, H_2\} = 0$. \hfill \qed

\section{Commuting discrete Lagrangians}\label{sect: discr results}

We start with two functions $\Lambda_i:M\times M\to\mathbb R$, $i=1,2$, and assume that they define symplectic maps $F_i: T^*M\to T^*M$,
$F_i(q_0,p_0)=(q_i,p_i)$, according to 
\begin{equation}\label{eq: single Lagr map}
   p_0=\frac{\partial\Lambda_i(q_0,q_i)}{\partial q_0},\quad
    p_i=-\frac{\partial\Lambda_i(q_0,q_i)}{\partial q_i}.
\end{equation}

\begin{definition}
We say that the discrete Lagrangians $\Lambda_1$, $\Lambda_2$ commute, if the following two functions coincide identically:
\begin{equation}\label{S12}
S_{12}(q_0,q_{12})=\min_{q_1\in M}\Big(\Lambda_1(q_0,q_1)+\Lambda_2(q_1,q_{12})\Big)
\end{equation} 
and
\begin{equation}\label{S21}
S_{21}(q_0,q_{12})=\min_{q_2\in M}\Big(\Lambda_2(q_0,q_2)+\Lambda_1(q_2,q_{12})\Big).
\end{equation} 
\end{definition}
\begin{theorem}\label{main theorem}
If the Lagrangians $\Lambda_1$, $\Lambda_2$ commute, then the maps $F_1$, $F_2$ commute:
\begin{equation}\label{eq: commute}
F_1\circ F_2=F_2\circ F_1.
\end{equation}
\end{theorem}

We can say that functions \eqref{S12}, \eqref{S21} are given by
\begin{align}
& S_{12}=\Lambda_1(q_0,q_1)+\Lambda_2(q_1,q_{12}), \label{S12 full}\\
& S_{21}=\Lambda_2(q_0,q_2)+\Lambda_1(q_2,q_{12}),  \label{S21 full}
\end{align}
provided $q_1,q_2$ are defined as the solutions of the following {\em corner equations}:
\begin{align}
\label{eq: E1}\tag{$E_1$}
&\frac{\partial\Lambda_1(q_0,q_1)}{\partial q_1}+
\frac{\partial\Lambda_2(q_1,q_{12})}{\partial q_1}=0,
\\
\label{eq: E2}\tag{$E_2$}
&\frac{\partial\Lambda_2(q_0,q_2)}{\partial q_2}+
\frac{\partial\Lambda_1(q_2,q_{12})}{\partial q_2}=0.
\end{align}

\begin{lem}\label{E1,E2 --> E,E12}
Let the Lagrangians $\Lambda_1$, $\Lambda_2$ commute. If $q_1$, $q_2$ satisfy corner equations \eqref{eq: E1}, \eqref{eq: E2}, then the following two corner equations are satisfied, as well:
\begin{align}
\label{eq: E}\tag{$E_0$}
&\frac{\partial\Lambda_1(q_0,q_1)}{\partial q_0}-
\frac{\partial\Lambda_2(q_0,q_2)}{\partial q_0}=0,
\\
\label{eq: E12}\tag{$E_{12}$}
&\frac{\partial\Lambda_1(q_2,q_{12})}{\partial q_{12}}-
\frac{\partial\Lambda_2(q_1,q_{12})}{\partial q_{12}} = 0.
\end{align}
\end{lem}
\begin{proof}
Upon differentiating \eqref{S12 full} and \eqref{S21 full} and taking into account corner equations \eqref{eq: E1}, \eqref{eq: E2}, we find:
$$
\frac{\partial S_{12}}{\partial q_0}=\frac{\partial\Lambda_1(q_0,q_1)}{\partial q_0}, \quad \frac{\partial S_{12}}{\partial q_{12}}=\frac{\partial\Lambda_2(q_1,q_{12})}{\partial q_{12}},
$$
and 
$$
\frac{\partial S_{21}}{\partial q_0}=\frac{\partial\Lambda_2(q_0,q_2)}{\partial q_0}, \quad \frac{\partial S_{21}}{\partial q_{12}}=\frac{\partial\Lambda_1(q_2,q_{12})}{\partial q_{12}}.
$$
The statement now follows from the assumption $S_{12}(q_0,q_{12})\equiv S_{21}(q_0,q_{12})$.
\end{proof}

\begin{figure}[tbp]
\centering
\subfloat[]{\label{Fig: corners1}
\begin{tikzpicture}[auto,scale=0.6,>=stealth',inner sep=2pt]
   \node (x) at (0,0) [circle,fill,label=-135:$q_0$] {};
   \node (x1) at (2,0) [circle,fill,label=-45:$q_1$] {};
   \node (x12) at (2,2) [circle,fill,label=45:$q_{12}$] {};
   \draw (x) to (x1) to (x12);
\end{tikzpicture}
}\qquad
\subfloat[]{\label{Fig: corners2}
\begin{tikzpicture}[auto,scale=0.6,>=stealth',inner sep=2pt]
   \node (x) at (0,0) [circle,fill,label=-135:$q_0$] {};
   \node (x2) at (0,2) [circle,fill,label=135:$q_2$] {};
   \node (x12) at (2,2) [circle,fill,label=45:$q_{12}$] {};
   \draw (x) to (x2) to (x12);
   \end{tikzpicture}
   }\qquad
   \subfloat[]{\label{Fig: corners0}
\begin{tikzpicture}[auto,scale=0.6,>=stealth',inner sep=2pt]
   \node (x) at (0,0) [circle,fill,label=-135:$q_0$] {};
   \node (x1) at (2,0) [circle,fill,label=-45:$q_1$] {};
   \node (x2) at (0,2) [circle,fill,label=135:$q_2$] {};
   \draw (x) to (x1);
   \draw (x) to (x2);
\end{tikzpicture}
}\qquad
\subfloat[]{\label{Fig: corners12}
\begin{tikzpicture}[auto,scale=0.6,>=stealth',inner sep=2pt]
   \node (x1) at (2,0) [circle,fill,label=-45:$q_1$] {};
   \node (x2) at (0,2) [circle,fill,label=135:$q_2$] {};
   \node (x12) at (2,2) [circle,fill,label=45:$q_{12}$] {};
   \draw (x1) to (x12);
   \draw (x2) to (x12);
\end{tikzpicture}
}
\caption{Four corner equations:  \protect\subref{Fig: corners1} $(E_1)$,  \protect\subref{Fig: corners2} $(E_2)$,  \protect\subref{Fig: corners0} $(E_0)$,  \protect\subref{Fig: corners12} $(E_{12})$.}
\label{Fig: corners}
\end{figure}
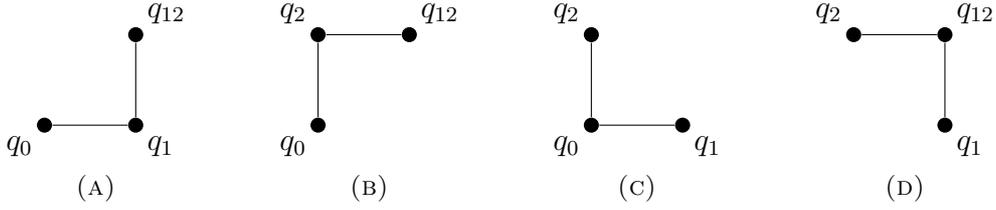

{\em Proof of Theorem \ref{main theorem}.} Take an arbitrary $(q_0,p_0)\in  T^*M$ and set $(q_1,p_1)=F_1(q_0,p_0)$ and $(q_{12},p_{12})=F_2(q_1,p_1)$. This means:
\begin{eqnarray}
\label{F1}
 p_0=\frac{\partial\Lambda_1(q_0,q_1)}{\partial q_0}, & \; & 
    p_1=-\frac{\partial\Lambda_1(q_0,q_1)}{\partial q_1},\\
 \label{F2 shifted}
 p_1=\frac{\partial\Lambda_2(q_1,q_{12})}{\partial q_1}, & \; &
    p_{12}=-\frac{\partial\Lambda_2(q_1,q_{12})}{\partial q_{12}}.
\end{eqnarray}
In particular, comparing the second equation in \eqref{F1} with the first equation in \eqref{F2 shifted}, we see that equation \eqref{eq: E1} is satisfied.

For the so defined $q_{12}$, solve equation \eqref{eq: E2} for $q_2$, and then set
\begin{equation}\label{aux1}
p_2=-\frac{\partial\Lambda_2(q_0,q_2)}{\partial q_2}=\frac{\partial\Lambda_1(q_2,q_{12})}{\partial q_2}.
\end{equation}
By Lemma \ref{E1,E2 --> E,E12} we conclude that equations \eqref{eq: E} and \eqref{eq: E12} are satisfied. By virtue of the first equation in \eqref{F1} and the second equation in \eqref{F2 shifted}, this implies:
\begin{equation}\label{aux2}
p_0=\frac{\partial\Lambda_2(q_0,q_2)}{\partial q_0}, \quad  p_{12}=-\frac{\partial\Lambda_1(q_2,q_{12})}{\partial q_{12}}.
\end{equation}
From \eqref{aux1}, \eqref{aux2} we conclude: $(q_2,p_2)=F_2(q_0,p_0)$ and $(q_{12},p_{12})=F_1(q_2,p_2)$. \qed

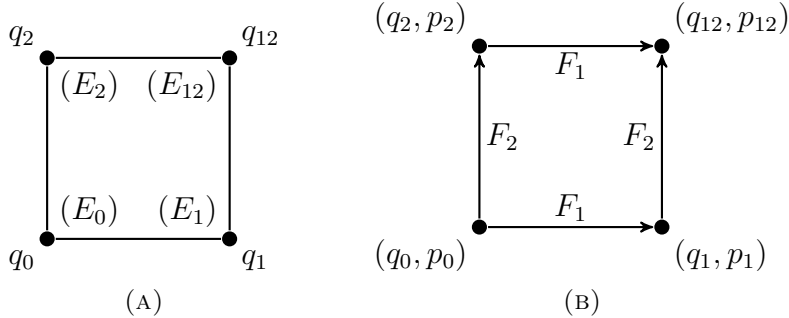
\begin{figure}[htbp]
\centering
\subfloat[]{\label{Fig: consistency1}
\begin{tikzpicture}[auto,scale=1.2,>=stealth',inner sep=2]
   \node (x) at (0,0) [circle,fill,thick,label=-135:$q_0$,label=45:$(E_0)$] {};
   \node (x1) at (2,0) [circle,fill,thick,label=135:$(E_1)$,label=-45:$q_1$] {};
   \node (x2) at (0,2) [circle,fill,thick,label=-45:$(E_{2})$,label=135:$q_2$] {};
   \node (x12) at (2,2) [circle,fill,thick,label=45:$q_{12}$,label=-135:$(E_{12})$] {};
   \draw [thick] (x) to (x1) to (x12);
   \draw [thick] (x) to (x2) to (x12);
\end{tikzpicture}
}\qquad
\subfloat[]{\label{Fig: consistency2}
\begin{tikzpicture}[auto,scale=1.2,>=stealth',inner sep=2]
   \node (x) at (0,0) [circle,fill,thick,{label=-135:$(q_0,p_0)$}] {};
   \node (x1) at (2,0) [circle,fill,thick,{label=-45:$(q_1,p_1)$}] {};
   \node (x2) at (0,2) [circle,fill,thick,{label=135:$(q_2,p_2)$}] {};
   \node (x12) at (2,2) [circle,fill,thick,{label=45:$(q_{12},p_{12})$}] {};
   \draw [thick,->] (x) to node {$F_{1}$} (x1);
   \draw [thick,->] (x) to node [swap] {$F_{2}$} (x2);
   \draw [thick,->] (x2) to node [swap]{$F_{1}$} (x12);
   \draw [thick,->] (x1) to node {$F_{2}$} (x12);
\end{tikzpicture}
}
\caption{ \protect\subref{Fig: consistency1} Corner equations attached to corners of an elementary square.  \protect\subref{Fig: consistency2} Maps $F_{1}$ and $F_{2}$ commute.}
\label{Fig: consistency}
\end{figure}

\textbf{Acknowledgements:} We are very happy to thank Nicolai Reshetikhin for many helpful discussions. This research is supported by the DFG Collaborative Research Center TRR 109 ``Discretization in Geometry and Dynamics''.

\end{document}